\documentclass[aps,prd,preprint,tightenlines,superscriptaddress,
amsfonts,amssymb,amsmath,showpacs,nofootinbib]{revtex4-2}

\usepackage{graphicx}%
\usepackage{multirow}%
\usepackage{amsmath,amssymb,amsfonts}%
\usepackage{amsthm}%

\usepackage{mathrsfs}%
\usepackage[title]{appendix}%

\usepackage{xcolor}%
\usepackage{textcomp}%
\usepackage{algorithm}%
\usepackage{algorithmicx}%
\usepackage{algpseudocode}%
\usepackage{colonequals}
\usepackage{listings}%

\newcommand{\bi}{\begin{itemize}}
\newcommand{\ei}{\end{itemize}}
\newcommand{\p}{\partial}

\def\be{\begin{equation}}
\def\ee{\end{equation}}
\def\bn{\begin{enumerate}}
\def\en{\end{enumerate}}
\newcommand{\nn}{\nonumber}
\newcommand{\bea}{\begin{eqnarray}}
\newcommand{\eea}{\end{eqnarray}}
\newcommand{\ba}{\begin{array}}
\newcommand{\ea}{\end{array}}
\newcommand{\bl}{\begin{align}}
\newcommand{\el}{\end{align}}

\newcommand{\g}{\gamma}

\newcommand{\kap}{\kappa}
\newcommand{\lam}{\lambda}
\newcommand{\ph}{\varphi}
\renewcommand{\th}{\theta}
\newcommand{\ep}{\varepsilon}
\newcommand{\eps}{\epsilon}
\newcommand{\bs}{\begin{subequations}}
\newcommand{\es}{\end{subequations} \noindent}
\newcommand{\half}{\tfrac{1}{2}}

\theoremstyle{thmstyleone}%
\newtheorem{theorem}{Theorem}
\newtheorem{proposition}[theorem]{Proposition}%

\theoremstyle{thmstyletwo}%
\newtheorem*{remark}{Remark}%
\newtheorem{corollary}{Corollary}

\theoremstyle{thmstylethree}%

\begin{document}

\title[Study on the BKL scenario]
{A study on the Belinski-Khalatnikov-Lifshitz scenario through
quadrics of kinetic energy}


\author*{Piotr P. Goldstein}\footnote{email: \texttt{piotr.goldstein@ncbj.gov.pl}}



\affiliation{Theoretical Physics Division National Centre for
Nuclear Research, Pasteura 7, Warsaw, Poland}


\date{\today}
\begin{abstract}
A detailed description of the asymptotic behaviour in the
Belinski-Khalatnikov-Lifshitz (BKL) scenario is presented through
a simple geometric picture illustrating the geometry of their
ordinary differential equations (ODE), which describe a
neighbourhood of the cosmic singularity. The Lagrangian version of
the dynamics governed by these equations is described in terms of
trajectories inside a conical subset of the corresponding space of
the generalised velocities. The calculations confirm that the
initial conditions of decreasing volume inevitably result in
eventual total collapse, while oscillations along paths reflecting
from a hyperboloid, similar to those predicted by Kasner's
solutions, occur on the way. The exact solution, found in our
previous work, proves to be the only one that shrinks to a point
along a differentiable path. Therefore, its instability means that
the collapse is always chaotic. It is also shown that the BKL
equations are not satisfied by the Kasner solutions exactly, even
in the asymptotic regime, although the precision of their
approximation may be high.
\end{abstract}
\maketitle
\tableofcontents

\section{Introduction}\label{sec1}
Application of the Einstein equations to cosmology posed questions
whether a cosmic singularity follows from these equations and (if
the answer be positive) what properties has the universe described
by their solution in the neighbourhood of the singularity. The
problem set by E.M. Lifshitz and I.M. Khalatnikov \cite{LK}, soon
joined by V.A. Belinski, set additional requirements for the
solution to be worth consideration: it should develop from a set
of nonzero measure on the manifold of initial conditions which has
the proper dimensionality, i.e., depends on the proper number of
arbitrary parameters. Moreover, the singularity should be physical
reality, rather than a result of simplifying assumptions (of a
model or a special frame of reference).

The first question concerns the sheer existence of such a
singularity. The initial approach in \cite{LK} gave the negative
answer to this question. Namely, the authors demonstrated
possibility of transformation to the synchronous frame of
reference (in which time is the proper time at each point) and
proved that all solutions of the required kind, whose the
determinant of the spatial metric tensor vanishes at a finite
time, are fictitious because their singularity can vanish in other
reference systems. However, later R. Penrose \cite{Penrose} and S.
Hawking \cite{Hawking} proved existence of a singularity
independent of the frame of reference, first for the case of a
collapsing star \cite{Penrose}, then for a class of universe
models \cite{Hawking}. This made the authors of \cite{LK}
reconsider their assertion. In the article of 1970 \cite{BKL}, BKL
considered possible generalisations of Kasner's model of the
homogeneous universe without matter.

The original Kasner's solutions describe the Euclidean metric
whose time dependence reads \cite{MTW}
\be\label{Kasner}
dl^2=t^{2p_1}dx_1^2+t^{2p_2}dx_2^2+t^{2p_3}dx_3^2
\ee
where
\bs
\be\label{p-cond1}
p_1+p_2+p_3=1 ~\text{ and}
\ee
\be\label{p-cond2}
p_1^2+p_2^2+p_3^2=1.
\ee
\es
Since one of these exponents has to be negative, Kasner's
solutions are singular at $t=0$. Moreover, they are anisotropic,
and their anisotropy, measured by ratios of scales in the
principal spatial directions, grows indefinitely on the approach
to the singularity (with two exceptions, up to exchange or the
indices, $p_1=-1/3,\,p_2=p_3=2/3$ and $p_1=p_2=0,\,p_3=1$). The
Kasner metric \eqref{Kasner} may be one of the possible answers to
another important question: on symmetries of the primordial
universe, close to the singularity. Although the recent universe
is isotropic, it does not determine that it has been isotropic
from the beginning. Therefore, cosmological models allowing for
primordial anisotropy should be taken into consideration.

The authors of \cite{BKL} look for a generalisation of the Kasner
solutions of \eqref{Einstein} to a possibly large class preserving
Kasner's properties: homogeneity and increasing anisotropy on the
approach to the singularity. The solutions should properly
describe dynamics of the universe in the neighbourhood of the
cosmic singularity. The asymptotic behaviour of the universe under
these assumptions is known as the BKL scenario. Limiting the
interest to the region close to the singularity allows for
reducing the Einstein equations to relatively simple ODE. Below,
we roughly summarise a derivation of these equations, as presented
in \cite{Ryan}. The derivation starts from the Einstein equations.
In natural units ($c=1,~G=1$)
\be\label{Einstein}
R_{\mu\nu}-\frac12\,R\, g_{\mu\nu}= 8\pi\,T_{\mu\nu},
\ee
where $\mu,\, \nu=0,1,2,3,~~ R_{\mu\nu}$ is the Ricci tensor, $R$
-- the Ricci scalar, $g_{\mu\nu}$ the metric tensor, $T_{\mu\nu}$
-- the stress-energy tensor (all the tensors symmetric upon
$\mu\rightleftarrows\nu$).

It was shown in \cite{BKL} that the contribution of matter, i.e.
the energy-momentum tensor $T_{\mu\nu}$ has higher order in time
on the approach to the singularity at $t=0$, compared to the
singular behaviour of the spacetime curvature; therefore it is
neglected in this description. For further calculations, the
authors of \cite{BKL} assume the universe to be homogeneous and
choose the synchronous frame of reference. This choice allows for
assuming the metric tensor in the form
\be\label{metric}
ds^2=dt^2-\g_{ab}(t)dx^a dx^b,
\ee
where $\g_{ab}$ are components of the spatial metric tensor in the
coordinate system $x^a,~a=1,2,3$; summation convention is used to
covariant-contravariant pairs of indices.

This metric, substituted to the Einstein equations, reduces the
4-dimensional (4D) problem to a problem of finding a 3D metric,
whose Lie algebra structure has been classified by Bianchi. The
further calculations correspond to Bianchi~IX; the structure
constants may be chosen as $\ep_{abc}$, where $\ep$ is the
Levi-Civita antisymmetric symbol.

The calculations look simpler if we rescale $dt$ by the factor of
the spatial volume, according to
\be
dt=\sqrt{\g} dt',
\ee
where $\g$ is the determinant of the spatial (time-dependent)
metric tensor $\g_{ab}$. In what follows we will omit the prime at
the rescaled-time symbol, bearing in mind that the singularity at
old $t=0$ now appears as the limit at $t=\infty$ for the new $t$,
so a collapse at $t\to\infty$ corresponds to expansion of the
universe starting at the old time equal to zero.

Of the ten Einstein equations for spacetime without matter, those
for the $^0_a$ components ($a=1,\, 2,\, 3$) provide only relations
between constants, they do not describe the dynamics. What remains
are six equations for $R_a^b,~~(a,b=1,2,3)$,~ and one for $R_0^0$.

Without the matter term, the six Einstein equations for the
spatial components read
\be
R_a^b=\frac1{2\g}\dot{\kap}_a^b+P_a^b=0\nn
\ee
and the one for the temporal component has the form
\be
R_0^0-R_a^a=\frac1{4\g}\left(\kap_a^b\kap_b^a-\frac{(\dot{\g})^2}{\g^2}\right)-P_a^a=0,\nn
\ee
where the dot over a symbol denotes differentiation with respect
to the rescaled time, $\kap_a^b=\g^{bc}\dot{\g}_{ca}$, while
$P_a^b$ are components of the 3D Ricci tensor, built from the
components of the spatial metric tensor and the structure
constants. The latter are chosen to be $\ep_{abc}$.

A possibility of integration of these equations follows from the
Bianchi identities, yielding
\be
\ep_{abc}\kap_a^b=C_c=const.,
\ee
where summation over identical indices applies, and $C_c$ is a
vector integral of motion.

The choice of the synchronous frame is not unique. E.g. we have
freedom of its rotation. We use this freedom for choosing the
constants of motion $C_c$ to fix its orientation. Our choice is
$C_1=C_2=0,\quad C_3\equalscolon C$

The spatial metric tensor $\hat{\g}$ may be diagonalised and
represented by its components along its principal axes in the 3D
space. Let the diagonalised matrix be
\be
\hat{\Gamma}=\mathrm{diag}(\Gamma_1,\Gamma_2,\Gamma_3)\,.
\ee
In general, the axes would rotate with respect to our chosen
frame. In terms of the Euler angles, the total rotation
$\hat{\mathcal{R}}$ is a composition (group product) of 3
rotations: by angle $\ph$ about the $z$ axis (precession angle),
by $\th$ about the $x$ axis (nutation angle) and by $\psi$ about
the~new $z$ axis (pure rotation angle)
\be
\hat{\mathcal{R}}=\hat{\mathcal{R}}_{\psi}\hat{\mathcal{R}}_{\th}\hat{\mathcal{R}}_{\ph}.
\ee
The spatial equations for the off-diagonal components of the Ricci
tensor determine the rotation
\begin{align}
\sin\,\th\,\sin\psi\,\dot{\ph}+\cos\,\psi\,\dot{\th}&
=\frac{\Gamma_2\Gamma_3\sin\,\th\,\sin\,\psi}{(\Gamma_2-\Gamma_3)^2}\nn\\
\sin\,\th\,
\cos\,\psi\,\dot{\ph}-\sin\,\psi\,\dot{\th}&=\frac{\Gamma_3\Gamma_1\sin\,\th\,\cos\,\psi}{(\Gamma_3-\Gamma_1)^2}\nn\\
\cos\,\th\,\dot{\ph}+\dot{\psi}&=\frac{\Gamma_1\Gamma_2\sin\,\th\,\cos\,\psi}{(\Gamma_1-\Gamma_2)^2}.
\end{align}
In addition, we have another $4$ equations: $3$ of them describe
the dynamics of three $\Gamma$'s; they correspond to the diagonal
spatial Einstein equations. The fourth one is the Einstein
equation for the temporal (diagonal) component $R_0^0$ (see
\cite{Ryan} for details). The assumption of indefinitely growing
anisotropy proves to be consistent with thus obtained equations.

Let us name $\Gamma_1$ the greatest of the $\Gamma$'s, the
smallest being $\Gamma_3$. Then, close to the singularity, we have
\be\label{Gamma-rel}
\Gamma_1\gg \Gamma_2\gg \Gamma_3\, .
\ee
Making use of this inequality, we can greatly simplify the
Einstein equations by neglecting terms of order
$\Gamma_3/\Gamma_2$, $\Gamma_2/\Gamma_1$ and higher.

In the zero order in these ratios, the r.h.s.'s of the equations
containing the $t$-derivatives of the angles vanish, which means
that the rotation of the principal axes stops on the approach to
the singularity: $(\th,\ph,\psi)\to(\th_0,\ph_0,\psi_0)$. Having
fixed the angles, we define
\be
\Gamma_1\equalscolon a,\quad \Gamma_2C^2\cos^2\th_0\equalscolon
b,\quad \Gamma_3C^4
\sin^2\th_0\,\cos^2\th_0\,\sin^2\ph_0\equalscolon c.\nn
\ee
Then the three equations corresponding to the diagonal spatial
components become
\begin{equation}\label{L1}
\frac{d^2 \ln a  }{d t^2} = \frac{b}{a}- a^2,~~~~\frac{d^2 \ln b
}{d t^2} = a^2 - \frac{b}{a} + \frac{c}{b},~~~~\frac{d^2 \ln c }{d
t^2} = a^2 - \frac{c}{b},
\end{equation}
subject to the constraint imposed by the temporal equation
\begin{equation}\label{L2}
\frac{d\ln a}{dt}\;\frac{d\ln b}{dt} + \frac{d\ln
a}{dt}\;\frac{d\ln c}{dt} + \frac{d\ln b}{dt}\;\frac{d\ln c}{dt} =
a^2 + \frac{b}{a} + \frac{c}{b} \, .
\ee
This ends the derivation according to \cite{Ryan}. Quantities
$\,a=a(t),\, b=b(t)$ and $\,c=c(t)$ are called directional scale
factors. They are, up to the multipliers of order $1$,
proportional to length scales in three principal directions of the
chosen synchronous reference system, while the time parameter $t$
is the proper time rescaled by the volume scale (remember that the
singularity at zero of the proper time corresponds to the limit
$t\to\infty$ if the initial state has zero volume). These
equations will be shortly called the BKL equations. Due to
time-reversibility of equations \eqref{L1}--\eqref{L2}, they may
describe both, expansion of the universe starting from the
singularity or its final collapse.

Numerous papers, were devoted to the analysis (both analytic and
numeric) of the asymptotic behaviour of the universe in the BKL
scenario, e.g. \cite{BKL,BKL3,Ryan,Belinski:2014kba}. A
Hamiltonian approach was analyzed in detail in \cite{CP}, and a
comparison with the diagonal mixmaster universe was done in
\cite{CMX}. The scenario was discussed in detail, on a broad
background of related Bianchi models, in the book by Belinski and
Henneaux \cite{book}. Compared to those analyses, the goal of this
paper is rather modest: we provide a simple geometric picture,
which yields some useful exact results on the BKL scenario,
staying within physics described by the BKL equations \eqref{L1},
\eqref{L2}. These equations are interesting in two aspects:
\bi\item
They provide information on the behaviour of the universe close to
the singularity.
\item
As they are approximate, it may be interesting, how their
solutions are related to the exact solutions, especially the
Kasner solutions.
\ei
Apparently simple nonlinear equations often reveal interesting
behaviour. This also applies to the system \eqref{L1}, \eqref{L2}.
So far, relatively little attention has been devoted to the
possibilities inherent in these equations themselves. Their exact
solutions \eqref{solution} were not known until the author's work
with Piechocki \cite{GP}. Recently, other explicit solutions were
found for special cases, where $a$ or $b/a$, or $c/b$ are equal to
zero \cite{Conte}. It seemed to be very likely that a more
detailed analysis of the equations \eqref{L1}, \eqref{L2}, could
yield interesting results about their solutions and consequently
on physics of the universe in the described regime. This analysis
is performed in this paper, with the stress on the asymptotic
behaviour near the singularity. It includes the question of the
oscillatory and chaotic character of the approach to the
singularity, which was first discussed in \cite{BKL69}. In this
paper, we provide a rigorous proof that the exact solution
\eqref{solution} is the only one which is differentiable down to
the limit $t\to\infty$. Since this solution was found to be
unstable in \cite{GP}, it means that the approach according to the
BKL equations is always chaotic.

This paper is structured as follows:

In section \ref{Sec2}, the earlier results are shortly summarised.
These include the basic properties of \eqref{L1}, \eqref{L2} and
the exact solution from \cite{GP}. Section \ref{Sec3} contains
description of methods, especially the geometric tool of the
present analysis, which is the cone of the kinetic part of the
Lagrangian (further called ``kinetic energy''). Section \ref{Sec4}
contains (in \ref{Sub4.2}) one of the main results, which is
uniqueness of the exact solution \eqref{solution} as the only one
which allows for the collapse of the universe with differentiable
dynamics of the length scales. In Section \ref{Sec5}, the
Kasner-like and quasi-Kasner solutions are described. The other
result, stating that the BKL equations are not satisfied by
asymptotics of the exact Kasner solutions, is discussed in
Subsection \ref{Sub5.2}.

Lengthy proofs have been put off to two appendices.

\section{Earlier results}\label{Sec2}
\subsection{Basic properties of the equations\label{2A}}
\indent\textit{Symmetries:} The way in which equations,
\eqref{L1}, \eqref{L2} were obtained determines that there is no
symmetry under permutation of $a,\,b$ and $c$ (see
\eqref{Gamma-rel}). On the contrary, the growing anisotropy
assumption results in $a\gg b\gg c$. The system is evidently
symmetric under time reversal $t\rightleftarrows -t$; thus it can
describe the universe in both a collapse or an explosion as its
reversal. The equations have two Lie symmetries \cite{P-proc}. The
first one is a shift in time $t\rightleftarrows t-t_0$ for any
$t_0$ (which is obvious for an autonomous system). The second is a
scaling symmetry: If $a,\,b$ and $c$ constitute the solution of
\eqref{L1}, \eqref{L2}, and $\lam$ is the scaling parameter,
$t'=\lam t,~~a'=a/\lam,~~ b'=b/\lam^3,~~c'=c/\lam^5$, then
$a',\,b'$ and $c'$ as functions of $t'$ make another solution of
the system \cite{P-proc}.

\textit{Dependence:} Apparently, the system is overdetermined, due
to the constraint \eqref{L2} imposed on solutions of \eqref{L1}.
However, the constraint specifies a value of the only constant of
motion. Therefore, each of the equations \eqref{L1} may be
obtained from a system consisting of the other two of \eqref{L1}
and the constraint \eqref{L2}. E.g. \cite{P-proc}, if we
substitute $\ddot{a}$ and $\ddot{b}$ from the first two of the
equations \eqref{L1} into the $t$-derivative of the constraint
\eqref{L2}, we obtain the 3rd equation of \eqref{L1} multiplied by
$(\dot{a}/a + \dot{b}/b)$ (the dot denotes time differentiation).
This way, the 3rd equation is shown to be dependent on the other
two of \eqref{L1} and the constraint \eqref{L2}, with the
exception of unphysical solutions satisfying $ab= const$ (which
would be consistent with equations \eqref{L1}, \eqref{L2} for
trivial $a=b=0$ only).

\textit{Canonical structure} \cite{CP} Substitution
\be\label{a2x}
a=\exp(x_1),\quad b=\exp(x_2),\quad c=\exp(x_3)
\ee
yields a system derivable from a Lagrangian
\be\label{Lagrx}
\mathcal{L}=\dot{x}_1\dot{x}_2+\dot{x}_2\dot{x}_3+\dot{x}_3\dot{x}_1+\exp(2x_1)+\exp(x_2-x_1)+\exp(x_3-x_2),
\ee
with the constraint \eqref{L2} turning into \cite{CP}
\be\label{consx}
\mathcal{H}\colonequals\sum_{i=1}^3\frac{\p\mathcal{L}}{\p
\dot{x}_i}\dot{x}_i-
\mathcal{L}=\dot{x}_1\dot{x}_2+\dot{x}_2\dot{x}_3+\dot{x}_3\dot{x}_1-\exp(2x_1)-\exp(x_2-x_1)-\exp(x_3-x_2)=0.
\ee
Equation \eqref{consx} clarifies the sense of the dependence
between \eqref{L1} and \eqref{L2}: the constraint \eqref{L2} is a
particular choice of the first integral $\mathcal{H}$ for
solutions of equations \eqref{L1}, namely $\mathcal{H}=0$.

The Lagrangian has a well defined potential and kinetic
``energies''. The latter is an indefinite quadratic form of
signature $(+,-,-)$, whose zero surface is a cone. As seen from
\eqref{consx}, the potential energy is always negative while the
total energy is zero. This means that the kinetic energy is
positive, i.e., position of the system in the space of
``velocities'' is inside the cone (further, the quotation marks
will be omitted, also for the accelerations, i.e., derivatives of
the velocities, as well as the kinetic, potential and total
energies).

\subsection{The exact solution}\label{ex-sol}
In \cite{GP}, we found an exact analytical solution of the BKL
equations. The solution is unique up to a time shift. It reads
\be\label{solution}
a(t)= \frac{3}{\lvert t-t_0\rvert },~~ b(t)= \frac{30}{ \lvert
t-t_0\rvert ^{3}},~~ c (t)= \frac{120}{\lvert t-t_0\rvert ^{5}} \,
\ee
where   $\lvert t - t_0\rvert  \neq 0$ and $t_0$ is an arbitrary
real number.

The exact solution may be obtained by a substitution of a Laurent
series about an arbitrary singular point $t_0$, or by looking for
a solution which has power-like behaviour for $t\to\infty$, or
else as a self-similar solution with respect to the scaling
symmetry mentioned in subsection \ref{2A} \cite{GP}.

\textit{Instability of the exact solution} In \cite{GP}, we
explicitly solved the linear equation for small perturbations of
the exact solution and we found that the exact solution is
unstable. The perturbations have two oscillatory components whose
amplitudes tend to zero as $t\to\infty$, but it is insufficient
for stability, as the scale factors $a,\, b,\, c$ also tend to
zero. The instability manifests in the growth of the ratios of the
perturbation amplitudes to the respective perturbed scale factors;
these ratios increase as $t^{1/2}$. A characteristic value of the
ratio between the oscillation frequencies (approximately equal to
2.06) is one of the results of \cite{GP}; some chance exists that
this ratio might have left marks in the spectrum of presently
observed waves.

If we consider expansion of the universe from a point, the
oscillations grow with time, but their growth is slower by the
factor $t^{-1/2}$ than the expansion. Hence, with respect to the
expansion itself, we can consider the expanding universe as stable
(the terms proportional to $K_3$ in equation (16) of \cite{GP},
though apparently increasing, represent time shift only).

\textit{Importance of the exact solution} The reader might
consider the exact solution unimportant: it requires very special
initial conditions and it is unstable. However, it will later be
proved that it is the only differentiable solution of the BKL
equations in which the collapse to zero occurs in all three
principal directions. The fact that the only solution suitable for
a model of the universe which smoothly collapses to a point (or
smoothly expands from a point) is unstable seems to be an
important property of the BKL scenario.

\section{Methods}\label{Sec3}
We apply the aforementioned Lagrangian formalism, and illustrate
the evolution of the system by its trajectory in the space of
velocities in the diagonalised version of Lagrangian \eqref{Lagrx}
(defined below, in the first subsection).

\subsection{Useful variables}
Transformation \eqref{a2x} naturally replaces the original
variables $a,\,b,\,c$ by their logarithms $x_1,\,x_2,\,x_3$,
suitable for the Lagrangian description. However, the description
becomes clearer if we diagonalise the kinetic energy. If we care
about simplicity of the equations rather than unitarity of the
diagonalising transformation (accepting its determinant to be
$-6$), a good substitution is
\be\label{x2u}
x_1=u_1-u_2-u_3,\quad x_2=u_1+2u_3,\quad x_3=u_1+u_2-u_3 \, ,
\ee
which yields the Lagrangian in the form diagonal in the velocities
$\dot{u}_1,\,\dot{u}_2,\,\dot{u}_3$,
\be\label{Lagru}
\mathcal{L}=3\dot{u}_1^2-\dot{u}_2^2-3\dot{u}_3^2+\exp\big(2(u_1-u_2-u_3)\big)+\exp(u_2-3u_3)+\exp(u_2+3u_3).
\ee
Variables $u_1,\,u_2,\,u_3$ define the principal directions in the
velocity space. The dynamics in the new variables is determined by
the Lagrange equations
\bs\label{Lagrequ}
\begin{align}
&\ddot{u}_1=\frac13 e^{2(u_1-u_2-u_3)},\label{Lagrequa}\\
&\ddot{u}_2=e^{2(u_1-u_2-u_3)}-e^{u_2}\cosh(3u_3),\label{Lagrequb}\\
&\ddot{u}_3=\frac13
e^{2(u_1-u_2-u_3)}-e^{u_2}\sinh(3u_3)\label{Lagrequc}.\\
&\text{with the constraint}\nn\\
\mathcal{H}\colonequals
&\,3\dot{u}_1^2-\dot{u}_2^2-3\dot{u}_3^2-e^{2(u_1-u_2-u_3)}-2e^{u_2}\cosh(3u_3)=0.\label{Lagrequcons}
\end{align}
\es

In terms of the original variables, the new ones are
\be\label{u2abc}
u_1=\frac13\ln(abc),\quad u_2=\frac12\ln(c/a),\quad
u_3=\frac16\ln(b^2/ac).
\ee
 As we can see, $u_1$ is the logarithm of the
volume scale, up to a multiplicative constant. Hence, the
orthogonalisation automatically separates dynamics of the volume
from that of the shape, thus doing what Misner introduced in the
first stage of his transformation for the mixmaster model
\cite{MTW}.

The velocities might simply be expressed in terms of the canonical
momenta $p_i=\p \mathcal{L}/\p \dot{u}_i,~ i=1,2,3$; then
$\mathcal{H}$ becomes the Hamiltonian, whose kinetic part is also
a diagonal quadratic form in the momenta. However, the momenta are
equal to the velocities, up to a multiplicative constant.
Therefore, we do not introduce extra momentum-variables.

The variables $u_1,\,u_2,\,u_3$ will be extensively used in our
further analysis.

Variables $a, \,b, \, c$ are not suitable for numerical
simulations, especially for their graphic presentation, because of
the disproportion between their sizes $a\gg b\gg c$. This purpose
is better served by quantities of equal order of magnitude. The
shape of equations \eqref{L1} and \eqref{L2} suggest that these
could be
\be\label{a2qrs}
q\colonequals a^2,\quad r\colonequals b/a,\quad s\colonequals c/b,
\ee
while their logarithmic counterparts
\be\label{qrs2y}
y_1\colonequals\ln{q},\quad y_2\colonequals\ln{r},\quad
y_3\colonequals\ln{s},
\ee
would be the counterparts of $x_1,\,x_2,\,x_3$ for the
corresponding Lagrangian description. A simple manipulation of the
original equations \eqref{L1} leads to those satisfied by the new
variables, which may be cast into a compact form
\be\label{L1qrs}
\left(\ba{c}\ln q\\ \ln r\\
\ln
s\ea\right)^{\!\centerdot\centerdot}=M\cdot\left(\ba{c}q\\r\\s\ea\right),
\text{~~or~~}\left(\ba{c}\ddot{y}_1\\ \ddot{y}_2\\
\ddot{y}_3\ea\right)=M\cdot\left(\ba{c}e^{y_1}\\ e^{y_2}\\
e^{y_3}\ea\right)
\ee
with the constraint given by
\be\label{L2qrs}
\frac12 \left(\ba{ccc}\ln q & \ln r &\ln s
\ea\right)^{\centerdot}\,\cdot
M^{-1}\cdot\left(\ba{c}\ln q\\ \ln r\\
\ln s\ea\right)^{\!\centerdot}-q-r-s=0
\ee
where the constant matrix $M$ is given by
\be\label{M}
M=\left(\ba{rrr}-2&\,2& 0\\2&-2&1\\0&\,1&-2\ea\right),\text{ with
 } \det M = 2, ~~ M^{-1}=\left(\,\ba{rrr}\tfrac32&\,2& 1\\2&2&1\\1&\,1&0\ea\,\right).
\ee
Equations \eqref{L1qrs} for $y_i,~i=1,\,2,\,3$ may be derived from
a simple Lagrangian
\be\label{Lagry}
\mathcal{L}= \frac12
\left(\ba{ccc}\dot{y}_1&\dot{y}_2&\dot{y}_3\ea\right)\cdot
M^{-1}\cdot\left(\ba{c}\dot{y}_1\\ \dot{y}_2\\
\dot{y}_3\ea\right)+e^{\,y_1}+e^{\,y_2}+e^{\,y_3}
\ee
The constraint again corresponds to $\mathcal{H}=0$, where
$\mathcal{H}$ differs from the Lagrangian \eqref{Lagry}, by the
opposite signs at the exponential functions. Explicitly
\be\label{consy}
\mathcal{H} =\frac34 \dot {y} _ 1^2 +
 2 \dot {y}_1 \dot {y} _ 2 + \dot {y} _ 2^2 + \dot {y} _ 2 \dot {y}_3 + \dot {y}_3\dot {y}_1
 -\left( e^{\,y_1}+e^{\,y_2}+e^{\,y_3}\right)=0.
\ee
Diagonalisation of the kinetic energy in the Lagrangian
\eqref{Lagry}, is achieved by substitution of $y_1,\,y_2$ and
$y_3$ with their values in terms of $u_1,\,u_2$ and $u_3$
respectively
\be\label{y2u}
y_1=2 (u_1 - u_2 - u_3),\quad y_2 = u_2 + 3 u_3,\quad y_3 = u_2 -
3 u_3,
\ee
which leads back to Lagrangian \eqref{Lagru} and the constrained
Lagrange equations which stem from it \eqref{Lagrequ}.

\subsection{The cone of kinetic energy}
Our basic geometric tool for analysis and presentation of the
dynamics will be the quadrics of kinetic energy.
\be\label{quadrics}
E_k\colonequals 3 \dot{u}_1^2-\dot{u}_2^2-3\dot{u}_3^2=\eps\ge 0.
\ee
For $\eps>0$ they are two-sheet hyperboloids, becoming a cone for
$\eps=0$.
 Assume that the initial conditions describe a universe, whose
volume is decreasing. In the variables $\dot{u}_1,\,\dot{u}_2$ and
$\dot{u}_3$, we have
\begin{proposition}
The dynamics of the universe which shrinks with $t$ takes place in
the lower interior of the cone
\be\label{l-i}
3\dot{u}_1^2-\dot{u}_2^2-3\dot{u}_3^2>0,\quad \dot{u}_1<0.
\ee
\end{proposition}
\begin{proof}
The first inequality (\textit{interior}) follows from the
constraint \eqref{Lagrequcons}, from which
$3\dot{u}_1^2-\dot{u}_2^2-3\dot{u}_3^2$ is equal to a sum of
exponential functions and hence it is positive. The second
(\textit{lower}, i.e. $\dot{u}_1<0$) is equivalent to the
assumption that the volume scale is decreasing, by the first
equation of \eqref{u2abc}.
\end{proof}

\begin{figure}\label{ccone1}
\begin{center}
\vspace{-2.5cm}
\includegraphics[width=1.2\textwidth]{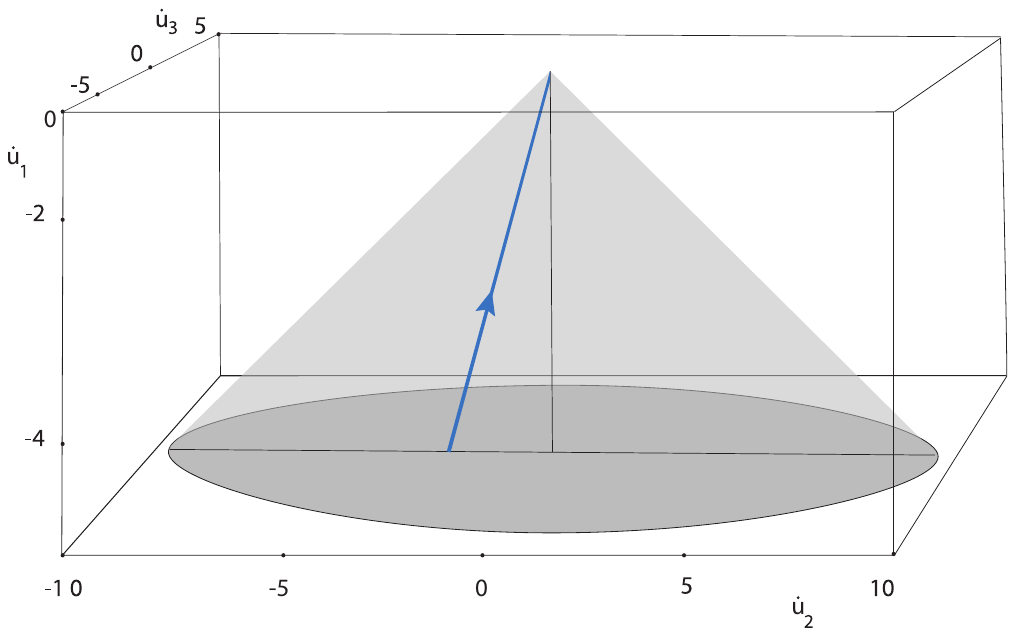}\vspace{-1.8cm}
\vspace{0.1cm} \caption{The lower half (=shrinking universe) of
the cone $3\dot{u}_1^2-\dot{u}_2^2-3\dot{u}_3^2
> 0$. The dynamics of the system takes place inside the cone.
The line with the arrow shows the exact solution; the arrow
indicates its direction of evolution. For $t\to\infty$, the line
tends to the apex of the cone. \\A position in the cone, together
with the tangent to the trajectory, provide complete information
on $u_1, u_2, u_3$, and their derivatives.}
\end{center}
\end{figure}

The following properties make the cone of kinetic energy a
particularly useful tool, reproducing essential information that
phase diagrams provide for single functions:
\begin{proposition}\label{singular}
The conical surface $ 3\dot{u}_1^2-\dot{u}_2^2-3\dot{u}_3^2=0$ is
a singular surface of the solution.
\end{proposition}
\begin{proof}
From the constraint \eqref{Lagrequcons}, if the kinetic energy
$E_k$ turns to zero, then the sum of exponential functions (the
minus potential energy, $E_p$) also has to be zero, whence all
exponents in \eqref{y2u} tend to $-\infty$ on approach to the
surface (including its apex). This requires that at least $u_1$
and $u_2$ tend to $-\infty$.
\end{proof}

\begin{proposition}
The position of the system in the cone, together with the
direction of the tangent to the trajectory, provide complete
information on the local values of $u_1,\,u_2,\, u_3$ and their
time derivatives.
\end{proposition}
\begin{proof}
The Cartesian coordinates of the position in the cone are the
components of the velocity, $\dot{u}_1,\,\dot{u}_2$ and
$\dot{u}_3$. The direction of the tangent yields proportions
between the components of the acceleration
$\ddot{u}_1,\,\ddot{u}_2$ and $\ddot{u}_3$. Given the components
of the velocity, the length of the acceleration vector can be
retrieved from
\be\label{length}
2\ddot{u}_2-9\ddot{u}_1+3\dot{u}_1^2-\dot{u}_2^2-3\dot{u}_3^2=0,
\ee
which is a simple linear combination of equations
\eqref{Lagrequa}, \eqref{Lagrequb} and \eqref{Lagrequcons} (with
the exception, $\ddot{u}_2/\ddot{u}_1=9/2$, which is possible only
on the conical surface). Having the accelerations, we can
calculate the values of $u_1,\,u_2$ and $u_3$ by solving the
system \eqref{Lagrequa}, \eqref{Lagrequb},
\eqref{Lagrequc} for these variables.\\
By differentiation of these equations, we can obtain higher
derivatives of $u_i ~~ i=1,2,3$ if they exist.
\end{proof}

\begin{proposition}\label{limits-dot}
Each of the velocities, $\dot{u}_1,\,\dot{u}_2$ and $\dot{u}_3$,
has a finite limit as $t\to\infty$.
\end{proposition}
\begin{proof}
Solving the dynamics equations \eqref{Lagrequa}, \eqref{Lagrequb},
\eqref{Lagrequc}, with respect to the exponential functions
(including the components of the hyperbolic ones), we get
\bs\label{for-exp}
\begin{align}
3\ddot{u}_1=&e^{2(u_1-u_2-u_3)},\label{for-expa}\\
4\ddot{u}_1-\ddot{u}_2-\ddot{u}_3=&e^{u_2+3u_3}\\
2\ddot{u}_1-\ddot{u}_2+\ddot{u}_3=&e^{u_2-3u_3}
\end{align}
\es
The r.h.s. of these equations are positive, whence their l.h.s.
are second derivatives of convex functions, and first derivatives,
of increasing functions,
$\dot{u}_1,~4\dot{u}_1-\dot{u}_2-\dot{u}_3$ and
$2\dot{u}_1-\dot{u}_2+\dot{u}_3$, respectively. The latter
functions are bounded, because the increasing property of
$\dot{u}_1$, together with $\dot{u}_1(0)<0$, infer
$\lvert\dot{u}_1(t)\rvert<\lvert\dot{u}_1(0)\rvert$, while both
$\lvert\dot{u}_2(t)\rvert$ and $\lvert\dot{u}_3(t)\rvert$ are not
greater than $\sqrt{3}\lvert\dot{u}_1(t)\rvert$ as long as we are
inside the cone. Hence, all three linear combinations of the first
derivatives are increasing functions bounded from above, and thus
have finite limits.

The determinant of the coefficient matrix for the linear
combinations in the l.h.s. of \eqref{for-exp} is nonzero ($=6$),
whence limits $g_1\!\colonequals \lim_{t\to\infty}\dot{u}_1,~
g_2\!\colonequals \lim_{t\to\infty}\dot{u}_2$ and
$g_3\!\colonequals \lim_{t\to\infty}\dot{u}_3$ may be uniquely
calculated from the limits of these linear combinations.
\end{proof}
\begin{corollary}
From this result, it follows that $g_1<0,~~\lvert g_2\rvert\le
\sqrt{3}\lvert g_1\rvert,~~\lvert g_3\rvert\le \lvert g_1\rvert$
and
\be
u_1\sim g_1 t,\quad u_2\sim g_2 t,\quad u_3\sim g_3 t \text{ as }
t\to\infty,
\ee
with the exception of $g_1=g_2=g_3=0$, corresponding to the apex
of the cone.
\end{corollary}
\begin{corollary}
As $\dot{y}_i,~~i=1,2,3$, are linear combinations of $\dot{u}_i$
\eqref{y2u}, they also have finite limits as $t\to\infty$.
However, for all $i$, $y_i\to -\infty$, which is a consequence of
the constraint \eqref{consy}.
\end{corollary}

\begin{proposition}\label{infinite}
A trajectory which ends on the surface or apex of the cone, needs
infinite time to reach it.
\end{proposition}
\begin{proof}
Consider a trajectory beginning in the lower half of the cone and
ending on its surface or apex. Since, $\dot{u}_1<0$, hence $u_1$
is a decreasing function of time. On this basis, time may be
calculated as
\be\label{time1}
t=\int_{u_1(0)}^{u_1}du_1'/\dot{u}_1'
\ee
 We have
$0>\dot{u}_1(t)\ge\dot{u}_1(0)$ in the lower half of the cone,
whence $1/\dot{u}_1(t)\le\dot{u}_1(0)<0$. Hence the integrand
$1/\dot{u}_1'$ is separated from $0$ in the interval of
integration. On the other hand, $u_1\to -\infty$ when we approach
the boundary (see the proof of Proposition \ref{singular}). The
integral \eqref{time1} in the limit $u_1\to -\infty$ extends over
infinite interval, while its integrand is separated from zero.
Hence, it is infinite.
\end{proof}
\begin{remark}
The time parameter $t$ calculated in \eqref{time1}, over a finite
or infinite interval, is always positive, as the integrand is
negative, while the lower limit of integration is greater than the
upper limit.
\end{remark}

\begin{proposition}
In the limit $t\to \infty$, each trajectory reaches the surface or
apex of the cone.
\end{proposition}
\begin{proof}
Time can also be expressed as
\be\label{time2}
t=\int_{\dot{u}_1(0)}^{\dot{u}_1}d\dot{u}_1'/\ddot{u}_1',
\ee
because $\dot{u}_1$ is an increasing function of $t$ (from
\eqref{Lagrequa}, commented in the proof of Proposition
\ref{limits-dot}). For any point of the lower interior of the
cone, the denominator is greater than zero (from
\eqref{Lagrequa}), whence the integrand is finite and so are the
limits of integration. Hence $t$ has a finite value. Merely for
$\left(\dot{u}_1,\,\dot{u}_2,\,\dot{u}_3\right)$ lying on the
surface or at the apex of the cone can $t$ become infinite.
\end{proof}
\begin{remark}
Variable $\dot{u}_1$ may in principle replace time as it is an
increasing function. Nevertheless, its use for this purpose is
limited, as its variation is very uneven (see Fig.\ref{ui}). There
are time intervals where the exponential function in
\eqref{Lagrequa} is close to zero and thus $\dot{u}_1$ hardly
increases; also the other exponential components in
\eqref{Lagrequ} are very small (see Fig. 2, 3), and the trajectory
comes very close to the surface of the cone. We call this
behaviour ``quasi-Kasner'' and discuss it in subsection
\ref{quasi}.
\end{remark}

\begin{proposition}
There is no possibility of a stop in the interior of the cone
(i.e., each point in the interior corresponds to nonzero
acceleration).
\end{proposition}
\begin{proof}
This property follows directly from equation \eqref{length}. As
long as $3\dot{u}_1^2-\dot{u}_2^2-3\dot{u}_3^2>0$, we have
$2\ddot{u}_2-9\ddot{u}_1<0$, which requires at least one nonzero
component of the acceleration.
\end{proof}
\begin{remark}
Note the absence of $\ddot{u}_3$ in \eqref{length}, which suggests
that solutions with $\ddot{u}_3=0$ may exist. Indeed, this is the
case of the exact solution \eqref{sol-u}.
\end{remark}
Situations where the trajectory in the velocity space slows down
to almost full stop may happen (see Fig. 2), which corresponds to
the ``quasi-Kasner'' behaviour discussed in Subsection
\ref{quasi}.
\section{Solutions ending in the apex of the cone}\label{Sec4}
\subsection{The exact solution in the cone}
In terms of the $u_i$ variables, the exact solution reads
\be\label{sol-u}
u_1=\frac13\ln\frac{10800}{\lvert t-t_0\rvert^9} ,\quad u_2=\half
\ln\frac{40}{\lvert t-t_0\rvert^4},\quad u_3=\frac16\ln\frac52.
\ee
Obviously, for a given sign of $t-t_0$, both $\dot{u}_1$ and
$\dot{u}_2$ have a simple pole, while $\dot{u}_3=0$. We also have
$\dot{u}_2=\frac23\dot{u}_1$, which means that the trajectory
corresponding to the exact solution is a half-line whose end lies
at the apex of the cone (see Fig. 1). Physically, it describes a
power-like collapse of all scale factors $a,\, b,\, c$ to zero,
i.e. a collapse of the universe, in all directions, to a point, as
$t\to\infty$ (which corresponds to the original time tending to
zero from the right).

The instability of the exact solution, found in \cite{GP} and
mentioned in subsection \ref{ex-sol}, affects also the solution in
terms of $u_i$, only the coefficients are different. However, the
solution itself is regular up to the apex.

\subsection{On the possibility of other solutions ending in the
apex}\label{Sub4.2} A question arises: are there any other paths
which approach the apex from the inner cone, along a regular (a
weaker assumption -- differentiable) curve, apart from that of the
exact solution?

The result is negative. Namely
\begin{proposition}\label{uni-apex}
The path in the cone, corresponding to the asymptotic of the exact
solution \eqref{sol-u}, i.e.
\be\label{ex-asymp}
\dot{u}_1\sim-\frac{3}{t-t_0},\quad
\dot{u}_2\sim-\frac{2}{t-t_0},\quad \dot{u}_3\sim 0,
\ee
is the only one which approaches the apex from the lower interior
of the cone along a differentiable curve.
\end{proposition}

The above result means that no other integral curve ending at the
apex may approach it at any definite angle with the axis of the
cone. Together with the instability of the exact solution, this
implies that all-direction collapse of the universe must be
chaotic.

The well-defined angle is equivalent to existence of finite limits
\be\label{lim-apex}
\lim_{t\to\infty}\ddot{u}_2/\ddot{u}_1\quad\text{and}\quad\lim_{t\to\infty}\ddot{u}_3/\ddot{u}_1.
\ee
To represent the collapsing universe, the trajectory should lie
within the lower interior of the cone, defined by the inequalities
\eqref{l-i}. This imposes another constraint on the values of the
limits \eqref{lim-apex}. Namely
\be\label{inner}
\lim_{t\to\infty}\left[\frac13\left(\frac{\dot{u}_2}
{\dot{u}_1}\right)^2+\left(\frac{\dot{u}_3}{\dot{u}_1}\right)^2\right]=
\lim_{t\to\infty}\left[\frac13\left(\frac{\ddot{u}_2}
{\ddot{u}_1}\right)^2+\left(\frac{\ddot{u}_3}{\ddot{u}_1}\right)^2\right]
\le 1.
\ee
The proof of the negative result is lengthy, therefore, it is put
off to Appendix A.

\section{Solutions ending on the surface of the cone}\label{Sec5}
From the fact that $\dot{u}_1$ increases, we conclude that it has
to approach the apex (as the exact solution) or the lateral
surface of the cone. In this section, we discuss the latter case.
\subsection{Asymptotics of the diagonal velocities}

Let a trajectory end on the conical surface, not at the apex.
Then, according to Proposition \ref{limits-dot}, all three
velocities have their limits, $\dot{u}_i\to g_i,~~i=1,2,3,~~g_1\ne
0$, which satisfy the equation of the cone
\be\label{g-cone}
3g_1^2-g_2^2-3g_3^2=0
\ee
With $\dot{u}_i\to g_i$, the asymptotic behaviour of the diagonal
variables is $u_i\sim g_i t$. Translating equation \eqref{g-cone}
into asymptotics of the scale factors, according to \eqref{u2abc},
we obtain
\be
a\sim\exp(2p_1 t),\quad b\sim\exp(2p_2 t),\quad c\sim\exp(2p_3 t)
\ee
where the common coefficient in front of $p_i,\,i=1,2,3$ might
have any value, depending on the time scale. By straightforward
calculation, equation \eqref{g-cone} turns into a constraint on
the constants $p_i$
\be\label{p-cond}
p_1p_2+p_2p_3+p_3p_1=0,
\ee
which by rescaling and choosing the direction of $t$ so that
$p_1+p_2+p_3=1$ (first Kasner's condition \eqref{p-cond1}
\cite{MTW}) is equivalent to the second Kasner's condition
\eqref{p-cond2}, in accordance with \cite{Ryan}.  This result
means that a solution whose trajectory ends on the conical surface
would behave as exact Kasner's solutions: the universe is squeezed
to zero in one direction while being stretched to infinity in the
remaining two.

\subsection{Impossibility of the exact Kasner-like
asymptotics}\label{Sub5.2} Later, we will see that the exact
Kasner solutions are reproduced with high precision by solutions
of the BKL. However, we are going to show that the aforementioned
exact Kasner-like solutions, though predicted and described in
\cite{Ryan}, cannot satisfy the BKL equations \eqref{L1},
\eqref{L2} (which is acceptable as these equations are
approximate). We will prove it in two stages, using the
$y_i,~i=1,2,3$ variables of \eqref{qrs2y} (which are connected
with the scale factors by \eqref{a2qrs} and with the $u_i$ through
\eqref{y2u}).

\underline{Stage 1}
\begin{proposition}\label{g12=0}
 Let $\g_1, \g_2, \g_3$ be
the limits of $\dot{y}_1,\dot{y}_2,\dot{y}_3$ (respectively) at
$t\to\infty$. Then $\g_1=\g_2=0$, while $\g_3\le 0$.
\end{proposition}
\begin{proof}
If for $t\to\infty$, the trajectory approaches the surface of the
cone, then the kinetic part in the constraint \eqref{consy} turns
to zero. In terms of the limits $\g_i$
\be\label{Eky=0}
\frac34 \g _ 1^2 +
 2 \g_1 \g _ 2 + \g _ 2^2 + \g _ 2 \g_3 + \g_3
 \g_1=0
\ee
The constraint \eqref{consy} requires that the potential part also
turns to 0. This means that the asymptotics $y_i= \g_i
t+o(t),~i=1,2,3$ has all $\g_i\le 0$. To also satisfy
\eqref{Eky=0}, the first two of the $\g$'s must be zero.
\end{proof}
\begin{corollary}
For trajectories ending at the surface but not at the apex, we
would have $\g_3<0$ (exactly), as vanishing of all three
$\dot{y}_i$'s corresponds to the apex.
\end{corollary}
\underline{Stage 2} Further limitation on $\g_i$ follows directly
from the dynamic equations \eqref{L1qrs} in their version
expressed in terms of $y_1,~y_2$ and $y_3$.
\begin{proposition}\label{g_i=0}
The only possible asymptotic behaviour of solutions to
\eqref{L1qrs}, which satisfies constraint \eqref{consy},
corresponds to $\g_1=\g_2=\g_3=0$.
\end{proposition}
\begin{remark}
This means that all solutions eventually end at the apex, i.e.,
the fate of the universe is a total collapse to a point (with
reversed time -- universe starts from a point). Moreover, together
with Proposition \ref{uni-apex}, it means that the collapse is
always chaotic.
\end{remark}
 The proof is lengthy and therefore it has been put
off to Appendix B.

\subsection{Quasi-Kasner solutions}\label{quasi}

Although exact Kasner solutions do not solve the system
\eqref{L1},\eqref{L2}, numerical calculations show that
approximate Kasner-like solutions of these equations are possible
and precise. Namely, the trajectories may approach the surface of
the cone and bounce at a short distance from it, thus switching
the universe to what may be considered the next Kasner epoch. The
trajectory then passes through the interior of the cone until it
approaches another point almost on its surface, at a less negative
value of $\dot{u}_1$ (as this coordinate may only increase,
according to \eqref{for-expa}). As the cone narrows, the amplitude
of this quasi-periodic oscillations diminishes. This behaviour
corresponds to reflections from the potential walls on Misner's
diagrams \cite{MTW}, while the surfaces of the corresponding
quadrics (lower halves of the two-sheet hyperboloids)
\be\label{hyper}
3 \dot{u}_1^2-\dot{u}_2^2-3\dot{u}_3^2=\eps_n
\ee
play the role of the equipotentials. The parameter $n$ indexes a
quadric at $n$-th reflection while $\eps_n$ is a measure of its
closeness to the surface of the cone.

Since the volume scale is proportional to $\exp{\left(\tfrac32
u_1\right)}$, while the time derivative, $\dot{u}_1$, is negative
throughout the evolution, the universe becomes more compact at
subsequent reflections, although the reduction need not concern
the scales in all directions.

Apparently, the velocity components $\dot{u}_i$ seem to remain
constant for some time and the kinetic energy looks as if it were
equal to zero. A logarithmic scale is necessary to reveal the
actual variation of these quantities, as may be seen on fig. 3.
This is due to the exponential dependence of the derivatives on
the values of these variables.

\begin{figure}\label{ui}
\begin{center}
\includegraphics[width=0.8\textwidth]{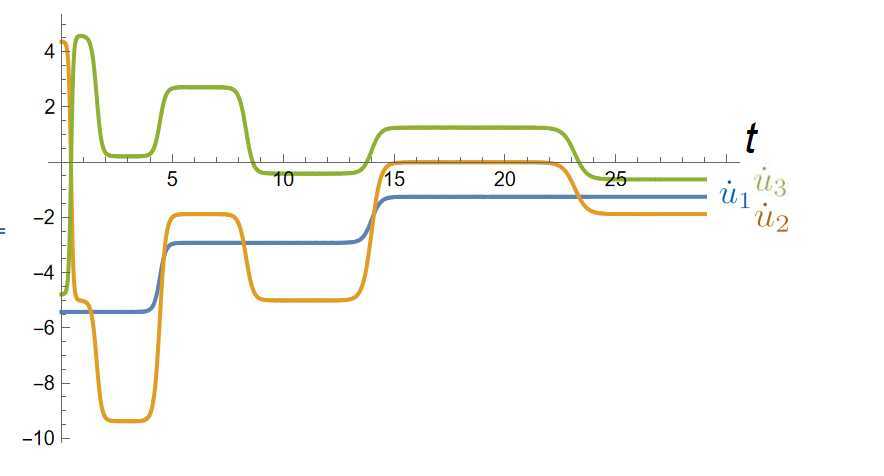}
\caption{Three components of the velocity $\dot{u}_1,\,\dot{u}_2$
and $\dot{u}_3$ as functions of time parameter $t$. Each of them
has time intervals of apparently constant values and there are
intervals in which all three seem to be constant. Revealing their
variability requires a logarithmic scale, as seen in the next
figure.}
\end{center}
\end{figure}

\begin{figure}\label{kin-en}
\begin{center}
\includegraphics[width=0.8\textwidth]{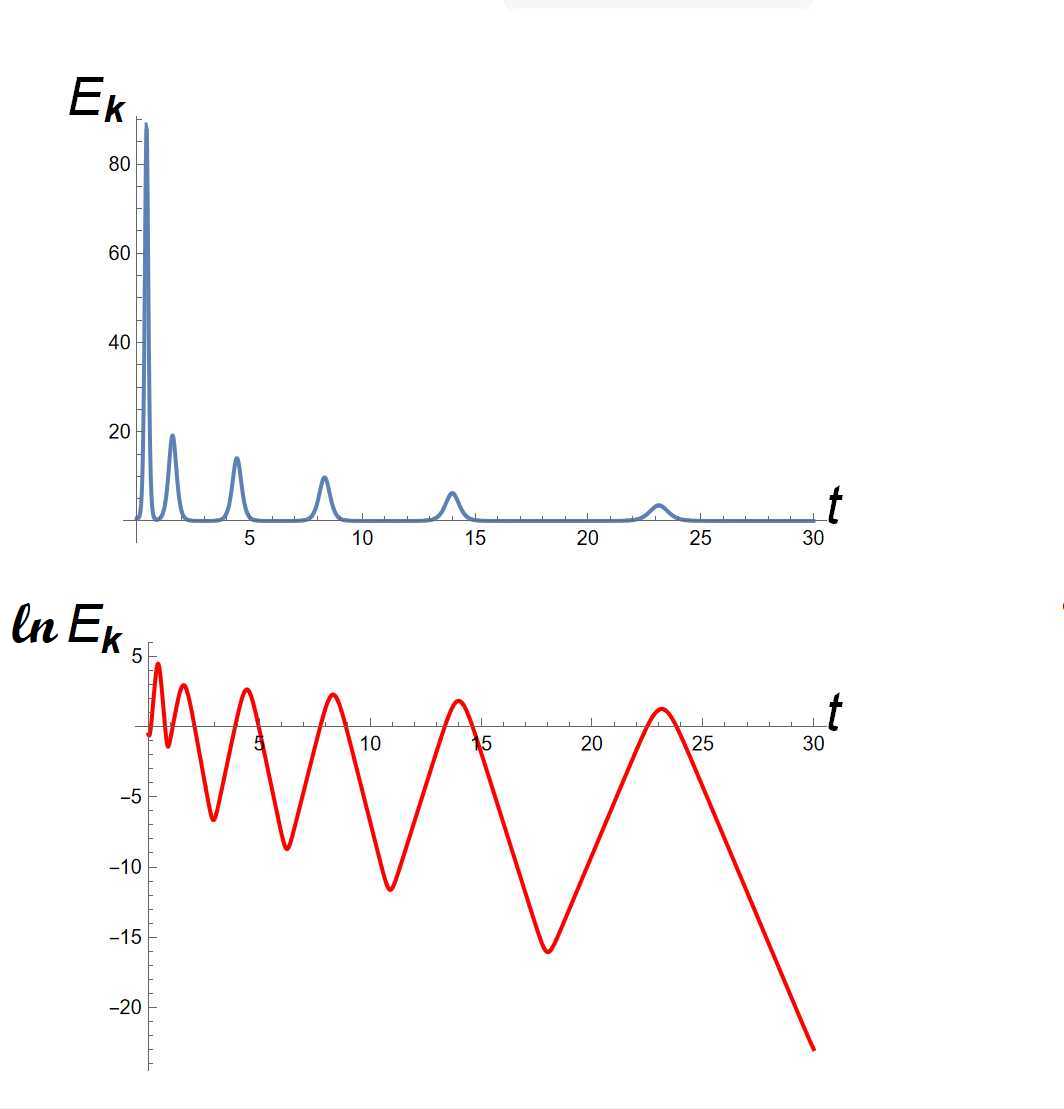}
\caption{The kinetic energy as a function of the time parameter
$t$. In the upper graph, apparently, $E_k$ systematically reaches
zero corresponding to the surface of the cone, and stays at this
level for a long time, but the logarithmic scale in the lower
graph reveals its oscillatory behaviour with reflections from the
hyperboloidal surfaces \eqref{hyper}.}
\end{center}
\end{figure}

\section{Conclusions}
\bi
\item
We have analysed details of the possible dynamics of the universe
while it undergoes contraction in the time parameter $t$,
according to the BKL equations \eqref{L1}, \eqref{L2}. Due to
reversibility of the equations, this may provide info on both
directions of the evolution for our universe. However, this
description is purely classic; it does not include quantum effects
like inflation.
\item
The asymptotics $t\to\infty$ may be oscillatory, but (unlike
predicted in \cite{P-proc}), limits at $t\to\infty$ exist for all
scale factors, their derivatives and first logarithmic
derivatives. However
\item
The reflections from the potential walls in the Misner's picture,
corresponding to reflections from a surface of a hyperboloid in
our picture, eventually lead to all-direction collapse of the
universe, due to the result of Proposition \ref{g_i=0}.
\item
The only asymptotic at $t\to\infty$ of this total collapse in
which the approach occurs along a differentiable path, it is the
exact solution \eqref{solution} (see Proposition \ref{uni-apex}).
This shows exceptional role of the exact solution, in spite of its
instability.
\item
The instability of the exact solution on the anisotropic approach
to collapse relies on growth of the ratio between the perturbation
and unperturbed scale factors (rather than the growth of the
perturbation itself). This conclusion was already presented in
\cite{GP}. Reversing the time arrow, we obtain the conclusion that
the universe would undergo isotropisation in a stable process.
\item
The final point at $t\to\infty$ is always the apex. Since the
exact solution is unstable, it cannot be an attractor. Bearing in
mind that the only approach to the apex along a differentiable
curve is that of the exact solution, we come up to the conclusion
that the generic approach to the limit has no limit of the ratios
between the scale factors or their derivatives. This infers that
it is chaotic.
\item
The cone in the velocity space, corresponding to zero kinetic
energy, together with the hyperboloids of constant $E_k$ has
proved to be a natural and useful tool for analysing the BKL
equations. This is likely to extend to other Lagrangian systems
with quadratic kinetic part.
\ei

\begin{appendices}

\section{Proof that the exact solution is the only differentiable approach to all-direction
collapse} (proof of Proposition \ref{uni-apex})
\begin{proof}
It is convenient to express \eqref{lim-apex} in terms of the $y_i$
variables, $i=1,2,3$ \eqref{y2u}, as each of them, as well as
their derivatives, have well-defined limits at the apex, namely
$y_i\to -\infty, \dot{y}_i\to 0$. For $i=2,3$, we assume existence
of two limits which define the paths of approach
\be\label{gii}
g_{i1}\colonequals \lim_{t\to\infty}
y_i/y_1=\lim_{t\to\infty}\dot{y}_i/\dot{y}_1=
\lim_{t\to\infty}\ddot{y}_i/\ddot{y}_1
\ee
(by de l'H\^opital's rule). The last pair of expression for the
limits define the direction of approach to the apex in the cone.
We are going to prove that their existence implies that the
approach to the apex is asymptotically identical with that of the
exact solution. In terms of the $y_i$ variables the latter reads
\be\label{sol-y}
y_1=\ln\frac{9}{(t-t_0)^2},\quad y_2=\ln\frac{10}{(t-t_0)^2},\quad
y_3=\ln\frac{4}{(t-t_0)^2}.
\ee

From equations \eqref{L1qrs}--\eqref{qrs2y}, for $i=2,3$, we have
\begin{align}\label{limitsy}
0\le
g_{21}=\lim_{t\to\infty}\frac{y_2}{y_1}=&\lim_{t\to\infty}\frac{\ddot{y}_2}{\ddot{y}_1}=\lim_{t\to\infty}\frac{2\exp{y_1}-2\exp{y_2}+\exp{y_3}}{2\exp{y_2}-2\exp{y_1}},
\nn\\
0\le
g_{31}=\lim_{t\to\infty}\frac{y_3}{y_1}=&\lim_{t\to\infty}\frac{\ddot{y}_3}{\ddot{y}_1}=\lim_{t\to\infty}\frac{\exp
{y_2}-2\exp{y_3}}{2\exp{y_2}-2\exp{y_1}}.
\end{align}
The nonnegative property of $g_{21}$ and $g_{31}$ follows from
$\forall_{i\in \{1,2,3\}}y_i\to-\infty$.

Consider the limit $\lim_{t\to\infty}2y_2/y_1+y_3/y_1$, first for
$g_{21}<1$
\be
0\le g_{213}\colonequals
2g_{21}+g_{31}=\lim_{t\to\infty}\frac{-3\exp{y_2}+4\exp{y_1}}{2\exp{y_2}-2\exp{y_1}}=\lim_{t\to\infty}\frac{-3+4\exp{[y_1(1-y_2/y_1)]}}{2-2\exp{y_1(1-y_2/y_1)}}=-\frac32.
\ee
If $g_{21}>1$, then, by an analogous transformation, $g_{213}=-2$.
Both results contradict the nonnegative property. Hence,
$g_{21}=1$.

For $g_{32}\colonequals \lim_{t\to\infty} y_3/y_2$, consider the
sum of reciprocals
\be\label{g231}
0\le g_{231}\colonequals
\frac{1}{g_{21}}+\frac{1}{g_{31}}=\lim_{t\to\infty}\frac{\exp{y_3}}{\exp{y_2}-2\exp{y_3}}.
\ee
If $g_{32} <1$, then, by a similar transformation,
$g_{231}\!\colonequals {1/g_{32}+1/g_{31}=-\frac12}$, which
contradicts the nonnegative property.

If $g_{32}>1$, then we obtain in a similar way, $g_{231}=0$.
Though consistent with the nonnegative property, this is
impossible since $g_{21}=1$ and $g_{31}\ge 0$ As a result, both
$g_{21}$ and $g_{31}$ must be equal to $1$. Substituting these
limits to \eqref{limitsy}, we obtain, after simple manipulation
\be\label{intermediate}
\lim_{t\to\infty}\frac{\exp y_3}{\exp y_2-\exp y_1}=4~\text{ and
}~\lim_{t\to\infty}\frac{\exp y_2}{\exp y_2-\exp y_1}=10,
\ee
which entails
\be\label{lim-dif-y}
\lim_{t\to\infty}(y_2-y_1)= \ln \frac{10}{9},\quad
\lim_{t\to\infty}(y_3-y_1)= \ln \frac49.
\ee
Finally, the asymptotic time dependence may be recovered from
\eqref{length}, which in terms of $y_i$ has the form
\be\label{lengthy}
\frac34\dot{y}_1^2+2\dot{y}_1\dot{y}_2+\dot{y}_2^2+\dot{y}_1\dot{y}_3+\dot{y}_2\dot{y}_3-\frac92\ddot{y}_1-5\ddot{y}_2-2\ddot{y}_3=0.
\ee
Dividing both sides of \eqref{lengthy} by $\dot{y}_1^2$ and
bearing in mind that all quotients $\dot{y}_i/\dot{y}_1$ and
$\ddot{y_i}/\ddot{y_1}$ tend to 1, we obtain the asymptotic, which
may be written as
\be
\lim_{t\to\infty}\frac{d}{dt}\left(\frac{1}{\dot{y}_1}\right)=-\frac12,
\ee
Integrating, we get the asymptotic of $y_1$ in the neighbourhood
of $t=\infty$
\be\label{y1}
\dot{y}_1= -2/(t-t_0),\qquad y_1= \ln \frac{C}{(t-t_0)^2}.
\ee
While the value of $t_0$ is arbitrary, the value of $C$ may be
recovered by substitution of \eqref{y1} into the constraint
\eqref{consy}, which yields $C=9$. Subsequent substitutions, of
this $C$ into \eqref{y1}, and the resulting $y_1$ into
\eqref{lim-dif-y}, yield precisely the asymptotic of the exact
solution.
\end{proof}
\section{Proof that all trajectories eventually tend to the apex}
(proof of Proposition \ref{g_i=0})
\begin{proof}
In Proposition \ref{g12=0}, we proved $\g_1=\g_2=0$, $\g_3\le 0$.
We are going to prove that $\g_3<0$ is impossible.

Assume $\g_3<0$. Adding first two equations in the right equation
of \eqref{L1qrs} (which corresponds to adding first two rows of
matrix $M$), we obtain
\be\label{y12}
\ddot{y}_1+\ddot{y}_2=\exp{y_3}.
\ee
Since  $\lim_{t\to\infty}\dot{y}_3=\g_3<0$, then for all $\ep>0$ a
time $T$ exists such that for all $t>T$, we have
\be
\dot{y}_3\in\, ]\g_3-\ep,~\g_3+\ep[,~ \text{ whence
}~y_3-y_3(T)\in \,](\g_3-\ep)(t-T),~(\g_3+\ep)(t-T)[\,.
\ee
Choose $\ep$ such that $\g_3+\ep<0$. We have
\be
\ddot{y}_1+\ddot{y}_2\in \left.
\right]e^{y_3(T)+(\g_3-\ep)(t-T)},~e^{y_3(T)+(\g_3+\ep)(t-T)}\left[
\right.\,,
\ee
with both exponents negative for large $t$. Hence, for these $t$
\be\label{dot-ineq}
\dot{y}_1+\dot{y}_2\in \left.
\right]\frac{1}{\g_3-\ep}e^{y_3(T)+(\g_3-\ep)(t-T)}+C_1,~\frac{1}{\g_3+\ep}e^{y_3(T)+(\g_3+\ep)(t-T)}+C_1\left[
\right.~,
\ee
where $C_1$ is a constant of integration. Since
$\g_1=\lim_{t\to\infty}\dot{y}_1=0$ and
$\g_2=\lim_{t\to\infty}\dot{y}_2=0$ as $t\to\infty$ (Proposition
\ref{g12=0}), we have $C_1=0$. Then, integrating again
\eqref{dot-ineq}, we obtain
\be
{y}_1+{y}_2\in \left.
\right]\frac{1}{(\g_3-\ep)^2}e^{y_3(T)+(\g_3-\ep)(t-T)}+C_2,~\frac{1}{(\g_3+\ep)^2}e^{y_3(T)+(\g_3+\ep)(t-T)}+C_2\left[
\right.~,
\ee
where $C_2$ is a constant of the next integration. However the
constraint \eqref{consy} requires that both ${y}_1\to -\infty$ and
${y}_2\to -\infty$ as $t\to\infty$, while the limit of the r.h.s.
is a finite number $C_2$. Hence, the assumption $\g_3<0$ leads to
a contradiction, whence $\g_3=0$.

The conclusion that all $\g_i,~~i=1,2,3$ are equal to zero means
that all trajectories eventually end at the apex.
\end{proof}
\end{appendices}

\end{document}